\documentclass[a4paper,UKenglish,cleveref,thm-restate]{lipics-v2021}

\pdfoutput=1
\hideLIPIcs

\graphicspath{{./figures/}}

\title{Drainability and Fillability of Polyominoes in Diverse Models of Global Control}
\titlerunning{Drainability and Fillability of Polyominoes}

\author{Sándor P. Fekete}{Computer Science, TU Braunschweig, Germany}{s.fekete@tu-bs.de}{https://orcid.org/0000-0002-9062-4241}{}
\author{Peter Kramer}{Computer Science, TU Braunschweig, Germany}{kramer@ibr.cs.tu-bs.de}{https://orcid.org/0000-0001-9635-5890}{}
\author{Jan-Marc Reinhardt}{Electrical Engineering and Computer Science, Bochum University of Applied Sciences, Germany}{jan-marc.reinhardt@hs-bochum.de}{https://orcid.org/0009-0005-8907-3832}{}
\author{Christian Rieck}{Discrete Mathematics, University of Kassel, Germany}{christian.rieck@mathematik.uni-kassel.de}{https://orcid.org/0000-0003-0846-5163}{}
\author{Christian Scheffer}{Electrical Engineering and Computer Science, Bochum University of Applied Sciences, Germany}{christian.scheffer@hs-bochum.de}{https://orcid.org/0000-0002-3471-2706}{}

\authorrunning{S.~P.~Fekete, P.~Kramer, J.-M.~Reinhardt, C.~Rieck, and C.~Scheffer}

\Copyright{Sándor P. Fekete and Peter Kramer and Jan-Marc Reinhardt and Christian Rieck and Christian Scheffer}



\funding{Work from TU Braunschweig and HS~Bochum was partially supported by the German Research Foundation~(DFG), project ``Space~Ants'', FE~407/22-1 and SCHE~1931/4-1.
	Work from University of Kassel was partially supported by DFG grant 522790373.
}

\acknowledgements{We appreciate the detailed feedback provided by the anonymous reviewers. We~thank Erik D. Demaine for early helpful conversations.}

\nolinenumbers

\usepackage[subrefformat=simple,labelformat=simple,position=b]{subcaption}

\usepackage{nicefrac}
\usepackage{mathtools}
\usepackage{xspace}
\usepackage[basic]{complexity}
\usepackage{todonotes}
\usepackage[ruled]{algorithm2e}
\LinesNumbered
\DontPrintSemicolon
\usepackage{tikz}
\usetikzlibrary{arrows.meta}
\usetikzlibrary{calc}
\usetikzlibrary{patterns}

\crefname{figure}{Figure}{Figures}
\crefname{theorem}{Theorem}{Theorems}
\crefname{lemma}{Lemma}{Lemmas}
\crefname{corollary}{Corollary}{Corollaries}
\crefname{section}{Section}{Sections}
\crefname{appendix}{Appendix}{Appendices}
\crefname{remark}{Remark}{Remarks}
\crefname{claim}{Claim}{Claims}
\crefname{conjecture}{Conjecture}{Conjectures}
\crefname{observation}{Observation}{Observations}

\newcommand{\ZZ}{\ensuremath{\mathbb{Z}^2}}
\newcommand{\Tpose}{\ensuremath{\mathrm{T}}}
\newcommand{\FT}{\ensuremath{\mathrm{FT}}\xspace}
\newcommand{\SSt}{\ensuremath{\mathrm{S1}}\xspace}
\newcommand{\IE}[1][\FT]{\ensuremath{{#1}_\mathrm{I}}\xspace}
\newcommand{\dual}[1]{\ensuremath{\overline{#1}}\xspace}
\newcommand{\rOne}[1][M]{\ensuremath{\rightarrow_{#1}}}
\newcommand{\rStar}[1][M]{\rOne[#1]^*}
\newcommand{\mSeq}[1][k]{\ensuremath{m_1,m_2,\ldots,m_{#1}}}

\newcommand{\mComp}[1][k]{\ensuremath{m_{#1}\circ\cdots\circ m_2 \circ m_1}}

\newcommand{\bScale}[2][k]{\ensuremath{{#2}^{\uparrow {#1}}}}
\newcommand{\gExt}[1][G]{\ensuremath{{#1}^{\leftrightarrow S}}}
\newcommand{\BigO}{\mathcal{O}}
\newcommand{\drainingProb}{\textsc{Maximum Tilt Draining Problem}}

\begin{document}
    \maketitle

    \begin{abstract}
      Tilt models offer intuitive and clean definitions of complex systems in which particles are influenced by global control commands. 
Despite a wide range of applications, there has been almost no theoretical investigation into the associated issues of filling and draining geometric environments.
This is partly because a globally controlled system (i.e., passive matter) exhibits highly complex behavior that cannot be locally restricted.
Thus, there is a strong need for theoretical studies that investigate these models both (1) in terms of relative power to each other, and (2) from a complexity theory perspective.
In this work, we provide (1) general tools for comparing and contrasting different models of global control, and (2) both complexity and algorithmic results on filling and draining.

    \end{abstract}

    \section{Introduction}
\label{sec:introduction}

The targeted use of global control mechanisms, applied synchronously and uniformly\footnote{Due to this intrinsic connection, we use the qualifiers ``uniform'' and ``global'' synonymously when applied to models of motion.} to small-particle matter (i.e., \emph{passive matter}), is both of great practical relevance and highest (theoretical) complexity.
A fundamental type of global control mechanism is that of uniform movement or translation, which finds application in a variety of manufacturing processes, such as filling polyhedra with a liquid for gravity casting~\cite{BoseT95,BoseKT98,YASUI2015494} or the intact removal of cast objects from their molds~\cite{bfghs2024singlepartmold,bhs2017separation}.
Global forces like electromagnetic fields and gravity play a key role in a range of further applications, such as amorphous computing and smart materials like smart paint~\cite{AbelsonACHHKNRSW00}, autonomous monitoring and repair systems as well as minimally invasive surgeries~\cite{bdflmw-particle2019}, medication~\cite{bflkkkrs-drugdelivery2020}, and biological robots~\cite{becker2013feedback}.

\pagebreak
Inspired by tilt games~\cite{abdelkader20162048,akitaya20212048,DemaineR18},
theoretical research commonly uses discrete, grid-based models to study the manipulation of passive matter by global control signals.
In this context, a distinction is made between single step and full tilt models~\cite{bdflmw-particle2019,caglsw-twodirections2023}, i.e., the movement of all particles by either one step or by a maximum distance in a uniform direction.

One of the most challenging problems remains the question of filling geometric shapes using global control:
Given a \emph{board}, defined as a subset of the square grid as well as a number of ``infill points'' (\emph{sources}), the question is whether, and if so how, the entire board can be filled by adding particles through the set of sources, with all particles moving in the same direction as determined by global control mechanisms.
Crucially, the inherently discrete nature of particle models introduces new complications that have not appeared in the continuous frameworks in~\cite{BoseT95,BoseKT98}, such as particles forming stacks by blocking each other instead of spreading like liquids.

Naturally, there exist boards that cannot be filled with a given number of sources.
We investigate the (therefore immediately arising) question of minimal necessary changes to the board in order to achieve fillability.
Even in simplified models, globally controlled particles and targeted changes to the board create dynamic systems that are extremely complex to analyze and control:
The smallest manipulation of a board can lead to far-reaching and not locally restrictable changes in terms of fillability, see for example Figures~\ref{fig:board}(b) and \ref{fig:board}(c).

Furthermore, it is unclear how the different models relate to each other in terms of the contrasting objectives of filling or draining.
In this paper, we provide a generalized, comprehensive framework for the comparative study of tilt models and
formally investigate what makes a given model \emph{more powerful} than another, i.e., allows more polyominoes to be filled or drained.
This includes the special case of single step and full tilt movement.

\subsection{Our Contributions}
\label{subsec:contribution}

We build upon established models of particle movement using global control
signals, namely the full tilt and single step models, examine them in the
context of generalized and more powerful models, and investigate how to achieve
drainability and fillability of polyominoes in these models by placing
obstacles. Our main contributions are twofold.

\begin{enumerate}[(1)]
\item We provide general tools to compare and contrast various models of
  particle movement, offering a more unified perspective and surprising new insights concerning the duality of different models (see \cref{fig:single-particle}). In particular, we prove the following:
  \begin{enumerate}[(1.1)]
	\item Equivalence between drainability in the full tilt model and fillability in the single step model (\cref{cor:s1-fill}).
	\item Limited relaxation of the restriction to global control signals does not tangibly affect drainability (\cref{thm:tc-fill}), i.e., does not increase model power in this regard.
  \end{enumerate}
\item With regards to gaining drainability via obstacles in the full tilt model, we provide:
	\begin{enumerate}[(2.1)]
		\item A reduction from a \textsc{3Sat} variant showing that it is \NP-hard to decide whether a given number of obstacles suffices, even when restricted to thin polyominoes, i.e., those containing no $2 \times 2$ squares (\cref{thm:hardness}).
		\item A $c$-approximation algorithm for $k$-scaled boards of $c=4$ for $k=3$ and of $c=6$ for $k>3$ (\cref{approx_ratio,approxratiogt}).
	\end{enumerate}
\end{enumerate}

\begin{figure}[htb]
\centering
\includegraphics{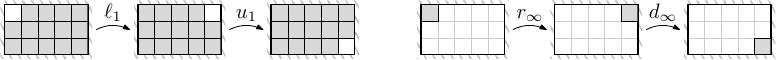}%
\caption{Movement of a single bubble in the single step model (left) is identical
  to the movement of a single particle in the full tilt model
  (right).}\label{fig:single-particle}
\end{figure}

For ease of exposition, we first present algorithmic results for
drainability in the full tilt model in
\cref{sec:preliminaries,sec:drain-full-tilt,sec:obstacles}, before generalizing
to fillability and other models in \cref{sec:models}.

\subsection{Related Work}
\label{subsec:relatedwork}

The drainability and fillability of polygonal two-dimensional shapes by uniform tilt movement has previously been studied by Aloupis et al.~\cite{acchlo-draining2014}, who gave an exact polynomial-time algorithm for the minimum number of sinks (i.e., exits) necessary to drain a polygon or polyhedron.
We study the discrete ``tilt'' model first introduced by Becker et al.~\cite{bhwrm-massive2013,caglsw-twodirections2023} in 2013.
Particles on a grid-based \emph{board} are moved uniformly, in either the single step or full tilt model:

The \emph{reconfiguration} problem asks whether a given arrangement of particles can be reconfigured into another specific arrangement, given a fixed set of obstacles.
The minimization variant, i.e., finding a shortest sequence of moves, was shown to be \PSPACE-complete in the full tilt model by Becker et al.~\cite{bdflmw-particle2019,6907856}.
A natural subproblem is the \emph{relocation} problem, which asks for just one specific particle to be moved to a target position.
In this variant, even deciding existence of any movement sequence is \NP-hard in the full tilt model~\cite{bdfhm-reconfiguring2013,bdflmw-particle2019}.
Recent complexity results for single step tilt moves by Caballero et al. demonstrate that, when restricted to two or three cardinal directions, deciding the existence of a relocation sequence is \NP-complete~\cite{ccglsw-hardness2020}, as well as in the absence of obstacles (i.e., a rectangular board) and just two movement directions~\cite{caglsw-twodirections2023}.
The more general \emph{occupancy} problem asks whether an arrangement of particles can be modified to move \emph{any} particle to a specific target position.
Even with this relaxation, finding a shortest sequence of moves, or deciding that no such sequence exists, is \PSPACE-complete~\cite{hierarchical2020,bdflmw-particle2019}.
Caballero et al.~\cite{caballero-cccg20-hardness} show that deciding occupancy remains \PSPACE-complete in the single step model when $2 \times 1$ dominoes are considered in addition to $1 \times 1$ particles.
However, it is easy to see that the problem is in \P\ if only unit-size particles are considered.

A myriad of related problems have been studied in the tilt model, such as \emph{gathering} particles into a connected configuration~\cite{bflkkkrs-drugdelivery2020,konitzny2022gathering}, or the design of special-purpose boards for efficient reconfiguration:
In particular, Balanza-Martinez et al.~\cite{balanza2019full} studied the design of universal shape constructors, i.e., boards that can be used to create large classes of particle configurations.
Further results exists on the reordering of labeled rectangular arrangements, which can be achieved either in linear time using quadratic area~\cite{bdflmw-particle2019}, or in quadratic time and linear area~\cite{ZhangCQB17}.
In addition to ``workspace''-based tilt reconfiguration settings, a variety of other models exist, e.g., moving only particles at maximal coordinates~\cite{aalr-trashcompaction2016,AkitayaLV23}, rather than all particles.
A~number of questions remain on the classification and complexity of deciding which configurations can be constructed~\cite{Becker:tilt_assembly,DBLP:journals/algorithmica/KellerRSS22} by sequentially introducing particles into an unobstructed system and ``gluing'' them onto an existing configuration using tilt movement.

    \section{Preliminaries}\label{sec:preliminaries}

Particles move on a \emph{board}, which we model as a finite subgraph of the
square tiling's dual graph, i.e., a board $B=(V,E)$ is a finite graph with
$V \subset \ZZ$ and \(E \subseteq \{\{u,v\} : {u \in V}, {v \in V},
\|u-v\|_1=1\}\). The \emph{boundary} of a board
is the axis-aligned polygon separating the tiles in~$V$ from ${\ZZ \setminus V}$.
Note that the boundary uniquely determines the corresponding board.
We only illustrate the special case of vertex-induced boards, which are isomorphic to sets of polyominoes and as such have boundaries whose sides only intersect at corners, even though the results in this paper hold in general.
A~\emph{sub-board} is a vertex-induced subgraph of a board; a maximally connected
\mbox{(sub-)board} is called a \emph{region}. We refer to the elements of $V$ as
\emph{pixels}. Every pixel is uniquely determined as the \emph{intersection} of
two \emph{segments}, a horizontal \emph{row segment} and a vertical \emph{column
segment}, which are maximally contiguous subsets of pixels of equal $y$-, or
$x$-coordinate, respectively. A \emph{boundary pixel} is a pixel adjacent to a
side of the boundary; a \emph{corner pixel} is a pixel adjacent to two
perpendicular sides of the boundary.
The set \(\{\ell, r, u, d\}\) is shorthand for the directions
$\ell$eft, $r$ight, $u$p, and $d$own. For a pixel $p$ let $p^{\ell}$, $p^r$,
$p^u$, and $p^d$ denote the left and right boundary pixel in its row segment,
and the top and bottom boundary pixel in its column segment, respectively; see
\cref{fig:board}(a).

\begin{figure}[htb]
\centering
\includegraphics{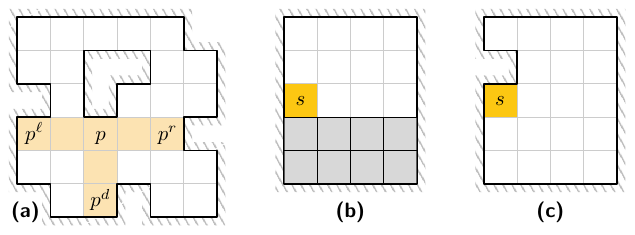}%
\caption{(a) A board with a pixel $p$ at the intersection of its row and column
  segments (shaded area), along with the boundary pixels of these segments. Note
  that $p=p^u$. (b) A configuration of a board $B$ that is minimal with respect to a
  sink $s$ in the full tilt model. Particles are depicted as gray squares. (c)~A~drainable sub-board
  of $B$.}\label{fig:board}
\end{figure}

A \emph{configuration}
$C \subseteq V$ is a subset of the pixels. We call a pixel $p \in C$
\emph{occupied} and a pixel $p \in V \setminus C$ \emph{free}. A \emph{move}
\(m: 2^V \to 2^V\) is a mapping between configurations; a \emph{model} is a set
of moves. Given a model $M$ and two configurations, $C$ and $D$, we say
\emph{$D$ is reachable from $C$ in one move}, or $C \rOne D$, if $D = m(C)$ for
some $m \in M$. $D$ is \emph{reachable from} $C$, or $C \rStar D$, if there is a
sequence of moves $\mSeq \in M$ such that $D = \mComp(C)$, i.e., $\rStar$ is the
reflexive, transitive closure of $\rOne$. We omit braces for singletons and
write, e.g., $u \rOne v$ instead of \(\{u\} \rOne \{v\}\). A pixel $p$ is
\emph{occupiable} from a configuration $C$ if there is a configuration $D$
reachable from $C$ that contains $p$.

The \emph{full tilt model}, \(\FT = \{u_\infty, d_\infty, \ell_\infty,
r_\infty\}\), has one move for every direction, which sees particles move
maximally in that direction until they hit the boundary or are blocked by
another particle. The move $\ell_\infty$, for example, transforms a
configuration $C$ in such a way that exactly the $|R \cap C|$ leftmost
pixels of every row segment $R$ are occupied in $\ell_\infty(C)$.

A \emph{sink} is a move associated with a pixel $s \in V$
defined as \(C \mapsto C \setminus \{s\}\). A configuration~$C$ is
\emph{minimal} with respect to a set of sinks $S \subseteq V$ if there is no
configuration $D$ such that $C \rStar[\FT \cup S] D$ and $|D| < |C|$. If \(V
\rStar[\FT \cup S] \varnothing\), then the board is \emph{drainable} to $S$.
\Cref{fig:fulltilt} illustrates particles moved to and removed at a sink.

Further notation is introduced as needed in individual sections.
\Cref{tbl:notation} provides a reference of frequently used symbols.

\begin{figure}[htb]
\centering
\includegraphics{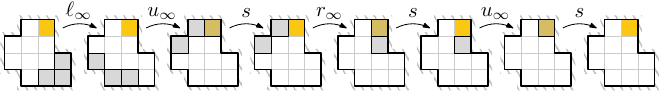}%
\caption{Draining three particles to a sink $s$ using full tilt
  moves.}\label{fig:fulltilt}
\end{figure}

\begin{table}
\centering
\caption{Frequently used notation.}\label{tbl:notation}
\begin{tabular}{ l l l }
 Notation & Meaning & Defined in \\
 \hline
 $C \rStar D$ & $D$ is reachable from $C$ in model $M$ & \Cref{sec:preliminaries} \\
 $\FT = \{u_\infty, d_\infty, \ell_\infty, r_\infty\}$ & the full tilt model & \Cref{sec:preliminaries} \\
 $G_F(B)$ & the full tilt graph of a board $B$ & \Cref{sec:drain-full-tilt} \\
 $G_S(B)$ & the small tilt graph of a board $B$ & \Cref{sec:drain-full-tilt} \\
 $\bScale{B}$ & the board $B$ scaled by a factor $k$ & \Cref{subsec:approx}\\
 $\gExt$ & the extended graph of $G$ with a set of sinks $S$ & \Cref{subsec:approx} \\
 $G_L(B,S)$ & the large tilt graph of a board $B$ and a set of sinks $S$ & \Cref{subsec:approx} \\
 $\SSt = \{u_1, d_1, \ell_1, r_1\}$ & the single step model & \Cref{sec:models} \\
 $\IE[M]$ & the interval extension of model $M \in \{\FT, \S1\}$ & \Cref{sec:models} \\
 $\dual{M}$ & the dual model of $M$ & \Cref{subsec:duality} \\
\end{tabular}
\end{table}

    \section{Drainability in the Full Tilt Model}\label{sec:drain-full-tilt}

In this section, we examine an algorithmic approach to deciding drainability in the full tilt model, 
starting with the following theorem.
This is a special case of a property that holds in more general classes of models, which is why we postpone the proof until \cref{drainable_mono_vp} in~\cref{sec:models}.

\begin{theorem}
    \label{thm:drainable-ft}
    A board $B=(V,E)$ is drainable to a set of sinks $S \subseteq V$ if and only if
    for every $p \in V$ there exists some $s \in S$ such that $p \rStar[\FT] s$.
\end{theorem}

The characterization in \cref{thm:drainable-ft} suggests an approach to
deciding drainability. Consider what we call the \emph{full tilt graph}
$G_F(B)=(V_F, E_F)$ of a board $B=(V,E)$. $V_F=V$ contains all pixels of $B$ and
\(E_F=\{(p,q) \in V^2 : p \neq q, q \in \{p^{\ell}, p^r, p^u, p^d\}\}\) connects
$p$ to $q$ if a particle at $p$ can reach $q$ in a single move. Deciding
drainability is equivalent to testing if there is a path in $G_F$ from every pixel to a sink.
However, this is unnecessarily inefficient. In~particular, we do not need to check every pixel
and require only the boundary as input.

\begin{lemma}
    \label{lem:corners-reachable}
    For every pixel $p$ on a board there is a corner pixel $c$ such that \(p
    \rStar[\FT] c\).
\end{lemma}
\begin{proof}
    Consider $w$ repetitions of the move $\ell \circ d$, where $w$ is the width
    of the board. After every application of $d$, the
    particle is positioned at a boundary pixel which is bottommost in its column
    segment. If $\ell$ does not move it, then it has reached a corner
    pixel. Otherwise, $\ell$ moves the particle to a pixel with lower
    $x$-coordinate, and there are exactly $w$ distinct $x$-coordinates.
\end{proof}

\begin{restatable}{corollary}{restateDrainCorners}
    \label{cor:drainable-ft-corners}
    A board $B=(V,E)$ is drainable to a set of sinks $S \subseteq V$ if and only if
    for every corner pixel $c \in V$ there exists some $s \in S$ such that $c \rStar[\FT] s$.
\end{restatable}
\begin{proof}
    Follows from \cref{thm:drainable-ft}, \cref{lem:corners-reachable}, and the fact that reachability
    is transitive.
\end{proof}

We now describe an algorithm that, given the boundary of a board $B=(V,E)$ and
a set $S \subseteq V$ of potential sinks, finds a minimum-cardinality subset
$S' \subseteq S$ such that $B$ is drainable to $S'$---or reports that no such
set exists. As a consequence, we can decide whether a board is drainable to a
set $S$ by supplying $S$ to the algorithm and observing if it returns a subset
or reports failure. The general approach is the same as the one used by the
authors of~\cite{acchlo-draining2014}. However, due to the rectilinear nature of
board boundaries, and our restriction to tilting in only four directions, our
algorithm requires less time.

\begin{figure}[htb]
    \centering
    \includegraphics{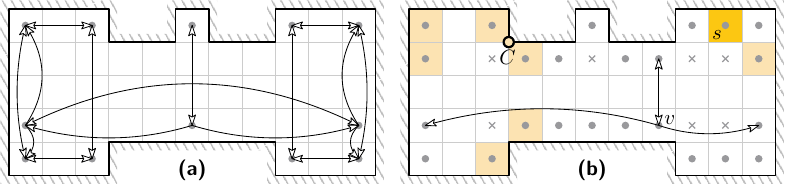}%
    \caption{(a) The small tilt graph, $G_S(B)$, of a board $B$. (b) The vertices of the large
    tilt graph, $G_L(B, \{s\})$, of $B$ and a sink $s$. The lightly-shaded pixels are added due to
    the reflex corner $C$ (though some are also part of the graph as corner
    pixels). Vertices marked with a cross are added as intersections. We only depict the
    edges incident on vertex $v$ to avoid clutter.}\label{fig:graphs}
\end{figure}

By \cref{cor:drainable-ft-corners}, the only pixels relevant for drainability are those
reachable from corner pixels. This allows us to restrict our algorithm to the \emph{small tilt graph},
$G_S(B)=(V_S, E_S)$, of~$B$. Let $V' \subseteq V$ be the set of corner pixels of
$B$. Then $V_S$ is the closure of $V'$ under $\rStar[\FT]$, and $G_S$ is the
subgraph of $G_F$ induced by $V_S$; see \cref{fig:graphs}(a).
For a boundary with $n$ corners, $|V_S| \in \BigO(n)$ as there are
at most two additional pixels for every reflex corner of the boundary, and
$|E_S| \in \BigO(n)$ as at most three other pixels are reachable from every
boundary pixel in a single move.

\begin{restatable}{lemma}{restateBuildGs}
    \label{lem:build-gs}
    $G_S(B)$ can be constructed from the boundary of a board $B$ in $\BigO(n \log n)$
    time and $\BigO(n)$ space, where $n$ is the number of corners of the boundary.
\end{restatable}
\begin{proof}
    We preprocess the boundary using the data structure by Sarnak and
    Tarjan~\cite{st-planar1986} to allow querying for the closest segment of the
    boundary to the left, right, up, and down of every point. Preprocessing requires
    $\BigO(n \log n)$ time and $\BigO(n)$ space, and queries take~$\BigO(\log n)$ time. This
    allows us to find, for every pixel $p$, the pixels $p^{\ell}$, $p^r$, $p^d$, and
    $p^u$ in $\BigO(\log n)$ time.

    We maintain a queue $Q$ and a set $D$ of discovered pixels, both initially
    containing all corner pixels. For every pixel $p \in Q$ we compute
    $p^{\ell}$, $p^r$, $p^d$, and $p^u$, add $p$ to $V_S$, add $(p, p^{\ell})$,
    $(p, p^r)$, $(p, p^d)$, and $(p, p^u)$ to $E_S$ (unless the head is equal to
    $p$), and extend $Q$ by \(\{p^{\ell}, p^r, p^d, p^u\} \setminus D\) and $D$ by
    \(\{p^{\ell}, p^r, p^d, p^u\}\). Setting up $Q$ and $D$ takes $\BigO(n)$ time and
    space, and processing $Q$ amounts to $\BigO(n)$ steps of $\BigO(\log n)$ time each.
\end{proof}

The next step is to consider the condensation $G_S^*=(V_S^*, E_S^*)$ of $G_S$,
i.e., the directed graph that has as vertices the strongly-connected components
of $G_S$ and an edge from component $A$ to component $B$ if there is an edge
from a vertex in $A$ to a vertex in $B$ in $G_S$. Then, $B$ is drainable to
every set $S$ that contains at least one pixel from every sink \(V' \in
V_S^*\). A suitable subset of sinks, if one exists, can thus be computed in $\BigO(|V_S|)$ time
and space. See \cref{alg:min_sinks} for an overview of the steps.

\begin{algorithm}[hbt]
    \KwIn{The boundary of a board $B=(V,E)$ and a set $S \subseteq V$.}
    \KwOut{A minimum-cardinality subset $S' \subseteq S$ such that \(V \rStar[\FT
    \cup S'] \varnothing\), or a message indicating no such set exists.}
    Construct $G_S$ using \cref{lem:build-gs}.\;
    Compute $G_S^*=(V_S^*, E_S^*)$, the condensation of $G_S$.\;
    $S' := \varnothing$\;
    \ForEach{sink $V' \in V_S^*$}{
        \eIf{there is at least one $s \in V' \cap S$}
        {\(S' := S' \cup \{s\}\)}
        {\KwRet{No solution exists because $V' \cap S = \varnothing$.}}
    }
    \KwRet{$S'$}\;
    \caption{Finding a minimum-cardinality set of sinks.}\label{alg:min_sinks}
\end{algorithm}

\begin{theorem}
    \label{lem:alg-min-sinks-correct}
    \Cref{alg:min_sinks} produces a correct result and runs in \(\BigO(n \log n +
    |S|)\) time using $\BigO(n)$ space, where $n$ is the number of corners of the
    boundary.
\end{theorem}
\begin{proof}
    As a consequence of \cref{cor:drainable-ft-corners}, $B$ is drainable to $S'$ if
    and only if there is a path from every $p \in V_S$ to an \(s \in V_S \cap
    S'\) in $G_S$. This is the case if and only if there is an $s \in V' \cap S'$
    for every sink \(V' \in V_S^*\). The algorithm is correct because it yields a
    smallest set~$S'$ with this property, if it exists. \cref{lem:build-gs} determines
    the time and space complexity of the first step. The following steps can all be
    completed in time and space linear in the size of $G_S$, which is $\BigO(n)$.
\end{proof}

    \section{Placing Obstacles to Guarantee Drainability}\label{sec:obstacles}

Although not every board is drainable, slight changes may render a board drainable, see Figures~\ref{fig:board}(b) and~\ref{fig:board}(c). 
Given a board and a set $S$ of sinks, we want to determine a maximum-size sub-board drainable to $S$. 
In other words, we want to place a minimum number of \emph{obstacles}
such that the remaining board is drainable.
We refer to this as the \drainingProb.

\subsection{Computational Complexity}

% !TeX root = ../00-main.tex
In this section, we show that the problem is \NP-hard; in particular, we show that deciding whether $\ell$ many obstacles suffice to guarantee that a board is drainable is an \NP-hard problem.
Our reduction is from \textsc{3Sat-3}~\cite{tovey1984simplified}, and works as follows; we refer to~\cref{fig:hardness-example} for an overview.

The problem \textsc{3Sat-3} is a variant of \textsc{3Sat} having the additional property that every variable appears at most $3$ times; thus, we may assume that each variable occurs at least once negated, and once unnegated.
Then, for every instance $\varphi$ of \textsc{3Sat-3}, we construct a polyomino $P_{\varphi}$ that is an instance of the \drainingProb.

\begin{figure}[htb]
	\centering
	\includegraphics[page=6, scale=1]{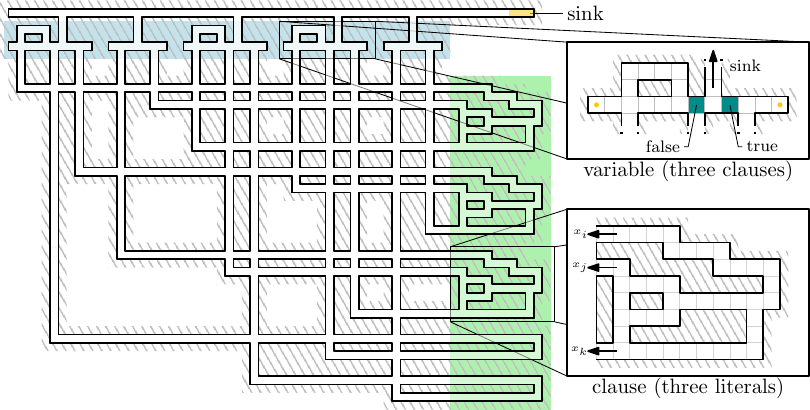}
	\caption{Overview of the \NP-hardness reduction for the full tilt variant. 
		The depicted instance is due to the \textsc{3Sat-3} formula $\varphi = (x_1 \lor \neg x_2 \lor x_3) \land (\neg x_1 \lor x_4 \lor \neg x_5) \land (x_2 \lor x_3 \lor \neg x_4) \land (x_1 \lor x_4) \land (\neg x_3 \lor x_5)$.
	}
	\label{fig:hardness-example}
\end{figure}

For this, we add for each variable the respective variable gadget. 
As illustrated in~\cref{fig:hardness-example}, the variable gadgets are placed in row at the top of the polyomino. 
For each clause, we place the respective clause gadget vertically in row at the right side of the construction. 
We then connect each variable with its respective clauses by thin L-shaped corridors.

It is easy to observe that we need at least one obstacle per variable to drain a variable gadget of $P_\varphi$.
Furthermore, from every other position we can reach a variable gadget, hence, no further obstacles are needed.
By carefully arguing we can show that $P_{\varphi}$ is drainable if and only if $\ell$ many obstacles are placed at very specific locations within the variable gadgets, where $\ell$ is the number of variables in $\varphi$.
This leads to the following theorem.

\begin{restatable}{theorem}{tiltFillabilityHardness}\label{thm:hardness}
	It is \NP-hard to decide whether placing $k$ obstacles suffices to guarantee drainability of a thin polyomino.
\end{restatable}

We specify variable and clause gadgets, and how these are connected to one another, and start with the variable gadget as depicted in~\cref{fig:hardness-variables}.

\begin{figure}[ht]
	\captionsetup[subfigure]{justification=centering}
	\begin{subfigure}[b]{0.43\columnwidth}%
		\centering
		\includegraphics[page=2]{./figures/tilt-fill-hardness}
		\caption{}
		\label{fig:hardness-variables_a}
	\end{subfigure}%
	\begin{subfigure}[b]{0.57\columnwidth}%
		\centering%
		\includegraphics[page=1]{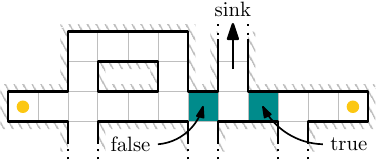}%
		\caption{}
		\label{fig:hardness-variables_b}
	\end{subfigure}%
	\caption{(a) Variable gadget for two occurrences, and (b) for three occurrences.}
	\label{fig:hardness-variables}
\end{figure}

\begin{lemma}
	\label{lem:obstacleForVariable}
	To drain the variable gadget, we need to place at least one obstacle.
	Moreover, placing one obstacle on one of two pixels within the gadget is sufficient to drain it, if their entrances are connected by a path.
\end{lemma}

\begin{proof}
	We first focus on the gadget in~\cref{fig:hardness-variables_a}.
	Assume for the following argument that both entrances at the bottom are connected by a thin path, and that the sink is in direction indicated by the arrow.
	
	Consider the left and right extreme pixel of a variable gadget, highlighted by yellow dots.
	From these pixels we cannot reach the sink; in particular, no other pixel than the respective other one is reachable from these pixels.
	It is easy to see that we have exactly two options, highlighted by green squares, to place obstacles in order to reach the sink from every pixel.
	
	Similar arguments apply for the gadget in~\cref{fig:hardness-variables_b}.
\end{proof}

As there are exactly two options for placing an obstacle in the desired manner within a variable gadget, we exploit these options for distinguishing between true and false.
Combining \cref{thm:drainable-ft,lem:obstacleForVariable} directly implies that an obstacle, placed somewhere outside of a variable gadget, has no influence on whether a particle can leave that variable gadget or not.
Additionally, as all variables are disjoint in our construction, we observe the following.

\begin{corollary}
	\label{lem:variablesDisjoint}
	For every variable gadget, we need a unique obstacle to drain it.
\end{corollary}

We proceed with the clause gadget, as depicted in~\cref{fig:hardness-clauses}.

\begin{figure}[htb]
	\captionsetup[subfigure]{justification=centering}
	\begin{subfigure}[b]{.5\columnwidth}%
		\centering
		\includegraphics[page=3,scale=0.9]{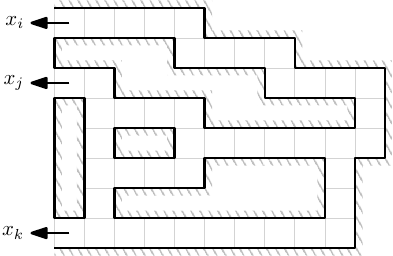}
		\caption{}
		\label{fig:hardness-clauses_a}
	\end{subfigure}%
	\begin{subfigure}[b]{.5\columnwidth}
		\centering
		\includegraphics[page=4,scale=0.9]{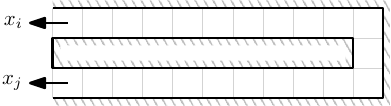}
		\caption{}
		\label{fig:hardness-clauses_b}
	\end{subfigure}%
	\caption{(a) shows a clause gadget that contains three literals, while (b) contains two literals.}
	\label{fig:hardness-clauses}
\end{figure}

\begin{lemma}
	\label{lem:clauseCorner}
	To reach the respectively contained variable gadgets from every corner pixel of a clause, we do not need to place any obstacles.
\end{lemma}

\begin{proof}
	Consider a particle on an arbitrary pixel within the clause gadget.
	It is easy to see that the particle can reach every exit of the gadget by a sequence of tilts.
	\Cref{fig:allpositionsClause} illustrates this by differently colored paths.
	
	\begin{figure}[ht]
		\centering
		\includegraphics[page=7]{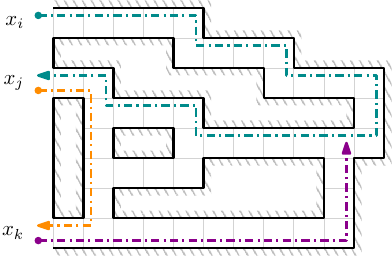}
		\caption{The respective clause gadget that contains three literals. The colored paths illustrate that all exits are path-connected to each other.}
		\label{fig:allpositionsClause}
	\end{figure}
	
	Because exits themselves are directly connected to the variable gadgets by simple L-shaped paths, a particle from every pixel of the clause gadget can reach all relevant variable gadgets, that is, all variable gadgets that contribute to this clause.
\end{proof}

By construction, a particle leaves a clause gadget horizontally, and reaches a convex corner pixel of the connecting path to the respective variable, from which the particle can either move back (and reaches a convex corner pixel of the clause gadget we came from), or it moves vertically into a variable gadget.
From there, we either move back, or horizontally; by the first alternative, the particle cannot reach any new pixels, and by the second it is trapped within the variable gadget, if there is no already placed obstacle.
Thus, we make the following observation.

\begin{observation}
	\label{obs:clauseToClause}
	Without placing any obstacles, a particle on any pixel of a clause gadget cannot move to any pixel of another clause gadget by any tilt sequence.
	The placement of obstacles within a variable also imply that a clause is only path-connected to another clause when they both contain the same literal.
\end{observation}

With this, we can prove~\cref{thm:hardness}, which we restate here.

\tiltFillabilityHardness*

\begin{proof}
	For every Boolean formula $\varphi$ that is an instance of \textsc{3Sat-3}, we construct a thin polyomino $P_\varphi$ as an instance of the \drainingProb.
	An~example is illustrated in~\cref{fig:hardness-example}.
	The construction is as follows:
	For every variable $x_i$, we place one variable gadget horizontally in a row at the top of the construction; we distinguish between variables that occur three times, and those who occur only twice.
	For every clause, we place the respective clause gadget vertically in a column at the right side of the construction, and connect every exit of the clause gadget via a thin L-shaped path to the respective input of its variable gadget.
	Note that we allow these paths to cross.
	It remains to connect all exits of the variable gadgets to a horizontal line; the sink is placed at its right end.
	
	We show that $P_\varphi$ is drainable if and only if there is a satisfying assignment for $\varphi$.
	
	\begin{claim}
		If $\varphi$ is satisfiable, then $P_\varphi$ with $\ell$ obstacles is drainable.
	\end{claim}
	
	Let $\alpha$ be a satisfying assignment for $\varphi$.
	For every variable $x_i\in \varphi$, we place the respective obstacle within the variable gadget as indicated in~\cref{fig:hardness-variables} according to its value in $\alpha$.
	By~\cref{lem:obstacleForVariable}, after placing these obstacles, we can reach the sink from every pixel within the variable, if their exits are connected by a path.
	This path is guaranteed via any clause gadget.
	Because $\alpha$ is a satisfying assignment, all corners of the clause gadgets in $P_\varphi$ are connected by a path to the sink via a respective variable gadget.
	Thus, by~\cref{thm:drainable-ft}, $P_\varphi$ is drainable.
	
	\begin{claim}
		If $P_\varphi$ with $\ell$ obstacles is drainable, then $\varphi$ is satisfiable.
	\end{claim}
	By~\cref{lem:obstacleForVariable}, we need to place at least one obstacle in each variable gadget.
	Because the number of variables is $\ell$, we in fact place exactly one.
	Thus, for every obstacle placed in each variable gadget, we set the respective value to the variables from $\varphi$.
	By the construction~$P_\varphi$ this is a satisfying assignment.
	If any clause was not satisfied, the corresponding clause gadget cannot reach the sink.
\end{proof}

\subsection{An Approximation Algorithm for Scaled Boards}\label{subsec:approx}

Our hardness proof relies on thin segments in the variable gadgets. To achieve
positive results, we restrict the considered boards to ones that do not have
thin segments by introducing scaling. In a \emph{$k$-scaled board} $\bScale{B}$
of a board $B$, every tile of the underlying square tiling gets replaced by a
$k \times k$ grid of tiles. This subdivides every pixel $p$ into $k^2$
\emph{sub-pixels}, for which $p$ is their \emph{super-pixel}. Sinks get scaled
as well, which means that the move associated with a sink may remove up to $k^2$
particles at once. By \cref{cor:drainable-ft-corners}, only the sub-pixels reachable
from corner sub-pixels are relevant in terms of drainability. We call these
\emph{outer} sub-pixels; the remaining sub-pixels are \emph{inner}
sub-pixels. Without any placed obstacles, only sub-pixels at a corner of their
super-pixel can possibly be outer sub-pixels, see \cref{fig:turn_gadget_3}(d).

The basic idea of our approach is to employ $k$-scaling, for $k \geq 3$, and use
inner sub-pixels as positions for obstacles, leading to new outer sub-pixels as intermediate steps on
paths from every pixel to a sink. We assume that a given board contains a
sink in every region---the only way to handle a region without a sink is to fill
it with obstacles. The chief benefit we get from obstacles are
new edges in the full tilt graph. We want as
few new edges as possible while guaranteeing the existence of paths from
every corner pixel to a sink.

For a directed graph $G=(V,E)$ and a set of sinks $S \subseteq V$, we define
the \emph{extended graph} \(\gExt=(V \cup \{r\} ,\gExt[E])\) to be the weighted
super-graph of $G$ that contains all edges $e \in E$ with weight zero and, for
every edge $(p,q) \in E$ with \((q,p) \notin E\), the \emph{inverse edge}
$(q,p)$ with weight one, in addition to edges $(s,r)$ with weight zero from
every sink $s$ to a newly added vertex $r$, called the \emph{root}. Note that
the extended graph contains an arborescence converging to $r$, as long as there
is a sink in every undirected component of the original graph.

Conceptually, we place \emph{turn gadgets} at sub-pixels of the heads of inverse
edges in a minimum-weight arborescence of the full tilt graph. For scaling
factor $3$, a turn gadget consists of two obstacles placed perpendicular to the
direction of the inverse edge in the middle of the pixel,
see Figures~\ref{fig:turn_gadget_3}(a)~to~(c). The central sub-pixel
in \cref{fig:turn_gadget_3}(c) gets disconnected by the obstacles from the rest
of the board and is not considered an outer sub-pixel, even though it is a
corner sub-pixel; it will be handled separately in the proof
of \cref{approx_feasible}.

\begin{figure}[htb]
\centering
\includegraphics{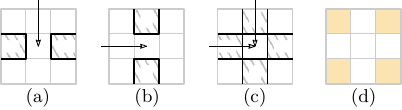}%
\caption{Turn gadgets and outer sub-pixels for scaling factor 3: (a) Horizontal
turn gadget determined by a vertical inverse edge. (b) Vertical turn gadget
determined by a horizontal inverse edge. (c) Horizontal and vertical turn
gadgets at the same pixel determined by two inverse edges. (d) Potential outer
sub-pixels.}\label{fig:turn_gadget_3}
\end{figure}

\paragraph*{The Large Tilt Graph}

Before we go into the details of the algorithm, we have to tackle one issue. As
was the case when deciding drainability in \cref{sec:drain-full-tilt}, computing
the full tilt graph from the boundary of a board is prohibitively expensive.
However, the small tilt graph developed for that purpose is insufficient in this
case, since it may contain undirected components without a sink.
To find a compromise that allows computing a feasible solution using a
reasonable amount of resources, we
introduce the \emph{large tilt graph} $G_L(B,S)=(V_L,E_L)$ of a board $B$ with
set of sinks $S$, which is another subgraph of the full tilt graph and a
super-graph of $G_S$. In~addition to the vertices of $V_S$, $V_L$ contains all
sinks \(s \in S\) and the pixels \(\{s^{\ell}, s^r, s^u, s^d\}\) for every sink
$s$.  Furthermore, for every reflex corner of the boundary between two segments,
the leftmost and rightmost (for row segments) or topmost and bottommost (for
column segments) pixels in those segments are included, as well as all pixels on
intersections of row segments and column segments containing an included
pixel. See \cref{fig:graphs}(b) for an example. $G_L$ is the subgraph of $G_F$
induced by $V_L$.

\begin{restatable}{lemma}{restateBuildGl}\label{lem:build-gl}
$G_L$ can be constructed from the boundary of $B$ and a set of sinks $S$ in
$\BigO(n \log n + K)$ time and $\BigO(n+K)$ space, where $n$ is the sum of the number
of corners on the boundary and the number of sinks, and $K$ is the
number of vertices arising from intersections.
\end{restatable}
\begin{proof}
    The setup and early phase is analogous to the construction of $G_S$ in
    \cref{lem:build-gs}: Preprocess, maintain queue $Q$ and set $D$, process pixels in
    the queue. However, building $G_L$ also requires handling sinks and reflex
    corners of the boundary, which is why they are also initially added to
    $Q$. Sinks are handled the same way as vertices of $G_S$ in \cref{lem:build-gs}.
    For every reflex corner $C$ encountered during the processing of $Q$, the
    leftmost and rightmost pixels of the row segments neighboring $C$, and the
    topmost and bottommost pixels of the column segments adjacent to $C$, are
    computed using the data structure built during preprocessing and added to $Q$
    and $D$, unless previously discovered. Thus, these steps still take $\BigO(n \log
    n)$ time and $\BigO(n)$ space.

    Additionally, we maintain a set of line segments $L$. Whenever two pixels $p$
    and $q$ on opposing ends of a row or column segment are discovered, either while
    handling a reflex corner or while adding pixels at the end of segments
    containing a sink or convex corner, the segment between them gets added to
    $L$. $\vert L \vert \in \BigO(n)$ because every corner and sink adds a constant number of
    line segments. After processing $Q$, intersections between segments in $L$ can
    be computed in $\BigO(n \log n + K)$ time and $\BigO(n+K)$ space using one of several
    known algorithms for finding line segment intersections, e.g., the one by
    Balaban~\cite{balaban-intersections1995}. Edges incident with intersections are
    computed in constant time per intersection, since they always have the
    intersection as tail and one of the endpoints of the intersecting segments as
    head. Thus, the whole process takes $\BigO(n \log n + K)$ time and $\BigO(n+K)$ space
    in aggregate.
\end{proof}

For every pixel $p$ there is a unique pixel $p^\dagger$, called the
\emph{anchor} of $p$, defined as the pixel from $G_L$ in the unique rectangle
containing $p$ adjacent on the inside to its upper-right corner and a single
pixel from the large tilt graph adjacent on the inside to its lower-left corner.
This is well-defined because two different rectangles with anchors $p_1^\dagger$
left of $p_2^\dagger$ would imply another pixel from $G_L$ in the rectangle
containing $p_1^\dagger$---either $p_2^\dagger$ itself, if $p_2^\dagger$ is not
below $p_1^\dagger$, or else the pixel at the intersection of the row segment of
$p_1^\dagger$ and the column segment of $p_2^\dagger$. Note that if $p$ is
in $G_L$, then $p=p^\dagger$.

\begin{lemma}
    \label{lem:anchor-commutes}
    For every pixel $p$ and every $x \in \{\ell, r, u, d\}$ the equality
    \((p^x)^\dagger = (p^\dagger)^x\) holds.
\end{lemma}
\begin{proof}
    Without loss of generality, assume $x=r$. $p^r$ is the rightmost pixel in its
    row segment, therefore $(p^r)^\dagger$ must be in the same column segment as
    $p^r$ and also a rightmost pixel in its row segment, since neighboring row
    segments only widen after reflex corners of the boundary. Furthermore,
    $(p^r)^\dagger$ must be in the same row segment as $p^\dagger$ (otherwise there
    would be an additional pixel of the large tilt graph in one of the rectangles
    due to an intersection). Thus, $(p^r)^\dagger$ is the rightmost pixel in the row
    segment of $p^\dagger$, i.e., $(p^\dagger)^r$.
\end{proof}

\begin{lemma}
    \label{lem:anchor-edge}
    For every edge $(p,q)$ of weight $w$ in the extended full tilt graph, there is
    an edge $(p^\dagger,q^\dagger)$ of weight $w$ in the extended large tilt graph.
\end{lemma}
\begin{proof}
    Without loss of generality, we consider the case that $p$ and $q$ are in the
    same row segment with $p$ left of $q$. If $(p,q)$ has weight zero, i.e., if it
    is in the original full tilt graph, then $q=p^r$. Thus, by using
    \cref{lem:anchor-commutes}, there is an edge from $p^\dagger$ to \(q^\dagger =
    (p^r)^\dagger = (p^\dagger)^r\) of weight zero. If $(p,q)$ has weight one, i.e.,
    if it is an inverse edge, then $p=q^\ell$. Again, by using
    \cref{lem:anchor-commutes}, there is an inverse edge from \(p^\dagger =
    (q^\ell)^\dagger = (q^\dagger)^\ell\) to $q^\dagger$, which has weight one.
\end{proof}

\begin{restatable}{observation}{restateGlConnected}\label{gl_connected}
The large tilt graph of a region is weakly connected.
\end{restatable}
\begin{proof}
    First, observe that $G_F$ of a region is weakly connected: Without loss of
    generality, consider two adjacent pixels $p$ and \(q=p + (1,0)^\Tpose\). Then,
    $p^r=q^r$ and there is an undirected path from $p$ to $q$ in $G_F$ via the edges
    $(p,p^r)$ and $(q,q^r)$. Thus, by induction on the length of a path, for every
    path in the region there is an undirected path in $G_F$. Now,
    by \cref{lem:anchor-edge}, and since for every inverse edge $(q,p)$ in the extended
    graph there is $(p,q)$ in the original, there is an undirected path between two
    pixels in $G_L$ whenever there is one in $G_F$.
\end{proof}

\begin{algorithm}[tb]
\KwIn{The boundary of a board $B=(V,E)$ and a set $S \subseteq V$ such that
  every region of $B$ contains an $s \in S$.}
\KwOut{A set of sub-pixels $O$ such that $\bScale{B}$ without $O$ is drainable
  to $S$.}
Construct $G_L$ using \cref{lem:build-gl}.\;
Give all edges of $E_L$ a weight of zero and add inverse edges $(q,p)$ of weight
one whenever $(p,q) \in E_L$ and \((q,p) \notin E_L\).\;
Add root $r$ and edges $(s,r)$ of weight zero for all $s \in S$.\;
Compute a minimum-weight arborescence $T$ converging to $r$ in the resulting
$\gExt_L$.\;
$O := \varnothing$\;
\ForEach{inverse edge $(q,p)$ in $T$}{
  Add the obstacles of the turn gadget determined by $(q,p)$ to $O$.\;
}
\KwRet{$O$}\;
\caption{Computing a drainable sub-board of $\bScale{B}$, for \(k \ge
  3\).}\label{alg:approx}
\end{algorithm}

\paragraph*{Analysis of the Algorithm}

We provide a listing of the steps in \cref{alg:approx} and proceed to analyze
its properties and performance. \Cref{fig:arborescence} illustrates an
arborescence in the extended large tilt graph of a board~$B$ with a single sink
and the resulting obstacle placement in $\bScale[3]{B}$.

\begin{figure}[htb]
\centering
\includegraphics{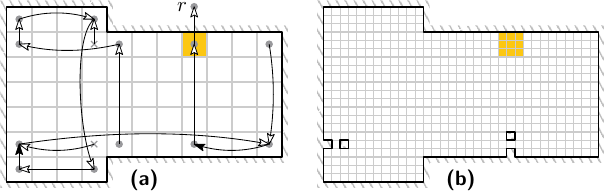}%
\caption{(a) An arborescence in the extended large tilt graph of a board $B$
with a single sink. Inverse edges are drawn with a black arrowhead. (b) The
resulting obstacle placement in $\bScale[3]{B}$.}\label{fig:arborescence}
\end{figure}

Before we prove that \cref{alg:approx} produces a drainable sub-board, it is
important to observe that certain (inverse) edges of the extended large
tilt graph cannot appear together in a minimum-weight arborescence of the large
tilt graph, or only under special circumstances.
\Cref{lem:one-vertical-gadget,common_head_sink} not only ensure the correctness
of a solution but also help bound the number of placed obstacles.

\begin{restatable}{lemma}{restateOneVerticalGadget}\label{lem:one-vertical-gadget}
For two distinct inverse edges $(p_1,q_1)$ and $(p_2,q_2)$ in a minimum-weight
arborescence of the extended large tilt graph with their heads $q_1$ and $q_2$
in the same row segment $R$, at most one of their tails $p_1$ and $p_2$ can be
in $R$.
\end{restatable}
\begin{proof}
    Assume, for sake of contradiction, that $p_1$ and $p_2$ are both in $R$. Note
    that $p_1$ and $p_2$ are the boundary pixels in $R$, i.e., edges $(p_1,p_2)$ and
    $(p_2,p_1)$ of weight zero exist in the extended large tilt graph. Without loss
    of generality, assume $q_1$ is at least as far away from the root in the
    arborescence as $q_2$. Then replacing $(p_1,q_1)$ with $(p_1,p_2)$ yields an
    arborescence of lower weight, contradicting the assumption that the arborescence
    has minimum weight.
\end{proof}

\begin{restatable}{lemma}{restateCommonHeadSink}\label{common_head_sink}
If there are two distinct inverse edges $(p_1,q)$ and $(p_2,q)$ in a
minimum-weight arborescence of the extended large tilt graph, then $q$ is a
sink.
\end{restatable}
\begin{proof}
    First, observe that exactly one of the edges $(p_1,q)$ and $(p_2,q)$ must be
    horizontal and one vertical, due to \cref{lem:one-vertical-gadget}. Further note
    that $q$ cannot be a boundary pixel because there are no inverse edges
    perpendicular to a boundary with a head at that boundary. Assume, for sake of
    contradiction, that $q$ is not a sink. Then there is a unique boundary pixel
    $q'$ such that $(q,q')$ is an edge in the arborescence. Neither $p_1$ nor $p_2$
    can be equal to $q'$, since an arborescence does not contain cycles. Therefore,
    $(q,q')$ must have the same orientation as either $(p_1,q)$ or
    $(p_2,q)$. Without loss of generality, assume it is $(p_1,q)$. Then replacing
    $(p_1,q)$ with $(p_1,q')$ yields a lower weight arborescence---a contradiction to
    the assumption that the arborescence has minimum weight.
\end{proof}

\begin{theorem}\label{approx_feasible}
\Cref{alg:approx} produces a drainable sub-board of
$\bScale{B}$, using $\BigO((n + K) \log n)$ time and $\BigO(n+K)$ space, where $n$ is
the sum of the number of corners of the boundary and the number of sinks, and
$K$ is the number of vertices of $G_L$ arising from intersections.
\end{theorem}
\begin{proof}
We first show that a sub-board produced by \cref{alg:approx} is drainable. Note
that \cref{alg:approx} assumes that every region of the board contains a sink,
which ensures that the extended large tilt graph of $B$ contains a path from
every pixel to a sink, by \cref{gl_connected}. By \cref{common_head_sink}, the
only pixels with two perpendicular turn gadgets are sinks. Therefore, particles
at all of their sub-pixels can be removed, including the one in the middle that
gets disconnected from the rest of the board by the turn gadgets,
see \cref{fig:turn_gadget_3}(c). The way the turn gadgets are constructed
guarantees that this is the only possible instance of a corner sub-pixel that is
not also an outer sub-pixel. By \cref{cor:drainable-ft-corners}, it now suffices to
show that for every outer sub-pixel $p$ there is a sub-pixel $s$ of a sink such
that \(p \rStar[\FT] s\), which we prove by strong induction on the distance $d$
from the super-pixel $p'$ of $p$ to a sink in the arborescence of the extended
large tilt graph.

For $d=0$, we can choose $s=p$, since $p$ is already a sub-pixel of a sink.

Now, assume $d>0$ and that every outer sub-pixel of a pixel closer than $d$ to a
sink in the arborescence has a sub-pixel of a sink reachable from it. Let $q'$
be the unique pixel such that $(p',q')$ is an edge in the arborescence. Without
loss of generality, assume $p'$ and $q'$ are in the same row segment $P$ with
$p'$ left of $q'$. By \cref{lem:one-vertical-gadget}, $P$ contains at most one
vertical turn gadget. We distinguish three cases.

\begin{enumerate}
\item There is no vertical turn gadget in $P$ or the only one is to the left of
$p$. Then, $q'=(p')^r$ and the move $r_\infty$ moves a particle at $p$ to a
rightmost outer sub-pixel of $q'$, see \cref{fig:approx_cases}(a).
\item A vertical turn gadget is positioned at $q'$. Then, the move $r_\infty$
moves a particle at $p$ to a leftmost outer sub-pixel of $q'$,
see \cref{fig:approx_cases}(b).
\item A vertical turn gadget is positioned at a pixel $v$ strictly between $p'$
and $q'$. Then, the move $r_\infty$ moves a particle at $p$ to a leftmost outer
sub-pixel of $v$, see \cref{fig:approx_cases}(c). The inverse edge $e$
determining the turn gadget at $v$ is either $((p')^\ell,v)$ or
$((p')^r,v)$. Either way, the path from the other boundary pixel to a sink in
the arborescence must include~$e$. Otherwise, $e$ could be replaced with an edge
between the boundary pixels for a lower weight arborescence. Because the edge from
$p'$ is to $q' = (p')^r$, $v$ lies on the path from $p'$ to a~sink.
\end{enumerate}
\begin{figure}[htb]
	\centering
	\includegraphics{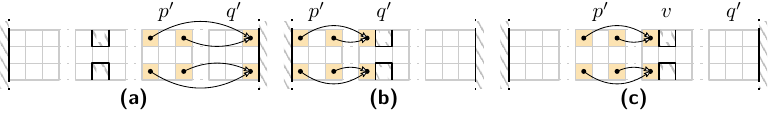}%
	\caption{The three cases in the proof of \cref{approx_feasible} for scaling
		factor $3$.}\label{fig:approx_cases}
\end{figure}

In all three cases, we reach an outer sub-pixel with a super-pixel whose
distance to a sink is strictly less than $d$. Thus, we can apply the induction
hypothesis to reach a sink sub-pixel.

Both remaining options, namely a vertical turn gadget at $p'$ or one to the
right of $q'$, lead to contradictions. The first would imply an
inverse edge $((p')^\ell,p')$, which could be replaced with \(((p')^\ell,
(p')^r)\) for a lower weight arborescence. In the second, $q'$ would not be
a boundary pixel and there would need to be another vertical turn gadget
at~$q'$, contradicting \cref{lem:one-vertical-gadget}.

As to the complexity, the initial steps can be completed in $\BigO(n \log n + K)$
time and $\BigO(n+K)$ space, by \cref{lem:build-gl}. $G_L$ has $\BigO(n+K)$ vertices and
edges, with $K \in \BigO(n^2)$, so the minimum-weight arborescence can be computed
in $\BigO((n+K) \log n)$ time and $\BigO(n+K)$ space using Tarjan's algorithm~\cite{tarjan-branchings1977}, which dominates the requirements of the remaining~steps.
\end{proof}

\begin{lemma}\label{lem:obstacle-weight-bound}
For a board $B=(V,E)$, a set of sinks $S \subseteq V$, and a set of obstacles
$O \subseteq V$ such that $B$ without $O$ is drainable to $S$, there is an
arborescence in the extended full tilt graph of $B$ of total weight at most
$2|O|$.
\end{lemma}
\begin{proof}
Consider the full tilt graph of $B$ and add, for every obstacle at a pixel
$p \in O$, the inverse edges \((p^\ell, p-(1,0)^\Tpose)\), \((p^\ell,
p+(1,0)^\Tpose)\), \((p^d, p-(0,1)^\Tpose)\) and \((p^d, p+(0,1)^\Tpose)\).
Call the resulting graph $H$. Now, for every two pixels $p$, $q$, with $(p,q)$
in the full tilt graph of $B$ without $O$, there is a path from $p$ to $q$ in
$H$: If $q$ is not next to an obstacle in $O$, then $(p,q)$ was in $G_F$ to
begin with. Otherwise, there is a path that first uses an edge in the initial
full tilt graph from $p$ to a boundary pixel $v$, followed by one of the added
inverse edges from $v$ to $q$.  Thus, $H$ is a subgraph of $\gExt_F$ that
contains a path from every pixel to a sink, i.e., a minimum-weight arborescence
in $\gExt_F$ has no larger weight than an arborescence in $H$. At~most one
outgoing inverse edge can be included per vertex, for a total weight of at most
$2|O|$ for any arborescence in $H$.
\end{proof}

Although we do all our calculations on the large tilt graph of the given board,
bounding the approximation ratio requires the weight of a minimum arborescence
of the extended full tilt graph of the scaled board.
\Cref{lem:arb-scaling,arb_full_large} together show that such an arborescence
cannot have smaller weight than the one we compute.

\begin{restatable}{lemma}{restateArbScaling}\label{lem:arb-scaling}
If there is an arborescence of weight $w$ in $\gExt_F(\bScale{B})$, for any
$k \geq 1$, then there is an arborescence of weight at most $w$ in $\gExt_F(B)$.
\end{restatable}
\begin{proof}
    Observe that whenever a sub-pixel $q$ is reachable from a sub-pixel $p$ in
    $\bScale{B}$, then the super-pixel $q'$ of $q$ is reachable from the super-pixel
    $p'$ of $p$ in $B$. Thus, choosing for each edge in an arborescence of
    $\gExt_F(\bScale{B})$ the edge between the corresponding super-pixels in
    $\gExt_F(B)$ yields a subgraph of $\gExt_F(B)$ of total weight at most $w$
    containing an arborescence.
\end{proof}

\begin{restatable}{lemma}{restateArbFullLarge}\label{arb_full_large}
If there is an arborescence of weight $w$ in $\gExt_F(B)$, then there is an
arborescence of weight at most $w$ in $\gExt_L(B,S)$.
\end{restatable}
\begin{proof}
    Take an arborescence of $\gExt_F$ of weight $w$ and replace every edge $(p,q)$
    with $(p^\dagger,q^\dagger)$. By \cref{lem:anchor-edge}, this leads to a valid
    spanning subgraph of $\gExt_L$ of total weight at most $w$. Since every vertex
    of $G_L$ is its own anchor, this graph retains an arborescence of $\gExt_L$ of
    weight at most $w$.
\end{proof}

\begin{theorem}\label{approx_ratio}
When applied with scaling factor $3$, \cref{alg:approx} places at most $4$
times as many obstacles as used in an optimum solution for $\bScale[3]{B}$.
\end{theorem}
\begin{proof}
Let $w$ be the number of inverse edges in a minimum-weight arborescence of the
extended large tilt graph of $B$ and $S$. Then, the number of obstacles placed
by \cref{alg:approx} is \(|\text{ALG}| = 2w\) because two obstacles are placed
per inverse edge, see \cref{fig:turn_gadget_3}. Let $|T(G)|$ denote the weight
of a minimum-weight arborescence in a graph $G$. The optimum number
$|\text{OPT}|$ of obstacles can be bounded in the following way.
\begin{align*}
  2|\text{OPT}| & \ge |T(\gExt_F(\bScale[3]{B}))|
    & \text{(\cref{lem:obstacle-weight-bound})} \\
  & \ge |T(\gExt_F(B))| & \text{(\cref{lem:arb-scaling})} \\
  & \ge |T(\gExt_L(B, S))| = w
    & \text{(\cref{arb_full_large})}
\end{align*}

Therefore, \(\frac{|\text{ALG}|}{|\text{OPT}|} \le 4\).
\end{proof}

\begin{figure}[htb]
\centering
\includegraphics{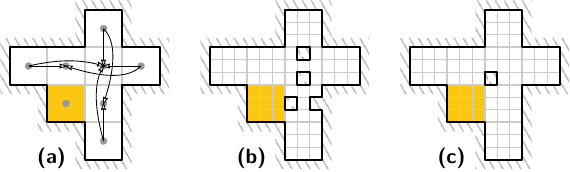}%
\caption{(a) The inverse edges in $\gExt_L(B)$ of a board $B$ with a single sink
$s \in S$. (b) Obstacles placed by \cref{alg:approx} for $\bScale[3]{B}$.
(c) An optimum obstacle placement for the same board and sink.}\label{fig:approx_tight}
\end{figure}

This is tight as there are $3$-scaled regions for which \cref{alg:approx} places exactly four times as many obstacles as required, see \cref{fig:approx_tight}.
For larger scaling factors we obtain the following:

\begin{restatable}{corollary}{restateApproxRatioGt}\label{approxratiogt}
When applied with scaling factor $k > 3$, \cref{alg:approx} places at most $6$
times as many obstacles as used in an optimum solution for $\bScale[k]{B}$.
\end{restatable}
\begin{proof}
    The proof is analogous to that of \cref{approx_ratio}, except that an additional
    obstacle is required per turn gadget, see \cref{fig:turn_gadget_k}, which
    increases the approximation ratio to 6.
\end{proof}

\begin{figure}[htb]
    \centering
    \includegraphics{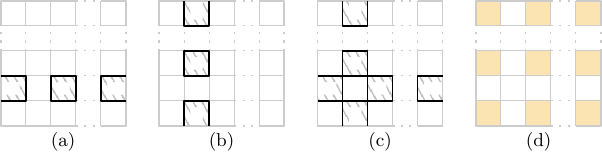}%
    \caption{Turn gadgets and outer sub-pixels for scaling factor $k > 3$:
    Horizontal (a) and vertical~(b) turn gadgets. (c) Both turn gadget at the same pixel. (d) Potential outer sub-pixels.}\label{fig:turn_gadget_k}
\end{figure}

We leave it as an open question whether our approach can be adapted to
$2$-scaling. Using the same ideas, we can ensure drainability of $2$-scaled
regions by placing thin walls instead of obstacles, i.e., by removing edges of
the underlying graph instead of vertices. It is not clear how obstacles can be
placed on $2 \times 2$ super-pixels to redirect particles coming from opposite
sides. An approximation may still be viable, but it will require a more
sophisticated argument that will certainly lose clarity.

Furthermore, note how two perpendicular turn gadgets as placed
in \cref{fig:turn_gadget_3}(c) split a connected board into two regions that
need to be drained individually. We can easily maintain connectivity by placing
another obstacle at the intersection of the turn gadgets---at the cost of
increasing the approximation ratio by $1$. A more elaborate question is whether
we can maintain the genus of the board, i.e., do not increase the number of
holes, while achieving the same approximation ratio. We leave this as an open
problem as well.

    \section{Drainability and Fillability in Generalized Models}\label{sec:models}

We now step away from the full tilt model and consider arbitrary models. For a model $M$ we
define the model $M^*$ to contain all compositions of moves from $M$, i.e., \(m
\in M^*\) if there are $\mSeq \in M$ such that $m = \mComp$, for any \(k \geq
0\). We generalize the concepts of minimality and drainability in the obvious
way. A move $m$ is \emph{monotone} if \(C \subseteq D\) implies \(m(C) \subseteq
m(D)\); it is \emph{volume-preserving} if \(|m(C)| = |C|\) for all $C$. Note
that the moves associated with sinks are monotone but obviously not
volume-preserving. Models are called monotone or volume-preserving if all of
their moves have the respective property.

The most well-studied model, apart from the full tilt model, is the
\emph{single step model}, \(\SSt = \{u_1, d_1, \ell_1, r_1\}\), which has
particles move to an adjacent pixel in one of the four directions, unless they
are blocked by the boundary or another particle. Formally, call a pixel
\emph{left-blocked} in a configuration $C$ if it and every pixel left of it in
its row segment are occupied in $C$. Then \(p \in \ell_1(C)\) if and only if $p$
is left-blocked in $C$ or \(p + (1,0)^\Tpose \in C\). The other moves are
defined analogously with respect to the other directions. Additionally, we
introduce an extension to the full tilt and single step models that allows
movement to be restricted to a subset of the segments parallel to the direction
of movement. Let $M \in \{\FT, \SSt\}$. Then the \emph{interval extension}
$\IE[M]$ has moves $m^{[i,j]} \in M \times \ZZ$ that apply the move $m \in M$ to
those segments whose $y$-coordinate (in the case of row segments and horizontal
moves) or $x$-coordinate (for column segments and vertical moves) lies in the
interval $[i,j]$, and leave all other segments of the affected type unchanged.
See the bottom half of \cref{fig:maximal1} for exemplary moves in \IE.
We start with an easily verified observation.

\begin{observation}\label{obs:concrete-properties}
The models \FT, \SSt, \IE, and \IE[\SSt] are monotone and volume-preserving.
\end{observation}

\begin{lemma}
    \label{lem:mono-confluence}
    For every monotone model $M$ and configurations $C$ and $D$ that are both
    reachable from $V$, there is a configuration $D' \subseteq D$ reachable from
    $C$.
\end{lemma}
\begin{proof}
    Let $D = m(V)$ for some $m \in M^*$. Then \(m(C) \subseteq m(V) = D\), due to
    monotonicity.
\end{proof}

\begin{proposition}
    \label{mono_mutual_size}
    In a monotone model, all minimal configurations reachable from $V$ are mutually
    reachable and have the same size.
\end{proposition}
\begin{proof}
    Let $M$ be a monotone model and $C$, $D$ configurations with $V \rStar C$ and
    $V \rStar D$. Then, by \cref{lem:mono-confluence}, there is $D' \subseteq D$ with
    $C \rStar D'$. Again, by \cref{lem:mono-confluence}, there is $D'' \subseteq D'$
    such that $D \rStar D''$. This implies $D = D' = D''$ because $D$ is
    minimal. Therefore, $C \rStar D$. Two mutually reachable, minimal configurations
    must have the same size, since otherwise the larger one would not be minimal.
\end{proof}

\begin{restatable}{proposition}{restateDrainableEmptyset}\label{drainable_emptyset}
A board is drainable in a monotone model if and only if $\varnothing$ is the
only minimal configuration.
\end{restatable}
\begin{proof}
    Follows directly from \cref{mono_mutual_size}.
\end{proof}

\begin{lemma}\label{minimal_vp}
A configuration $C$ is minimal with respect to a set of sinks $S$ in a
volume-preserving model if and only if no $s \in S$ is occupiable from $C$.
\end{lemma}
\begin{proof}
Let $M$ be a volume-preserving model, $S$ a set of sinks and $C$ a
configuration. First, assume there is a sink $s \in S$ and a configuration $D$
reachable from $C$ such that $s \in D$. Then~\(|s(D)| < |D| = |C|\), i.e., $C$
is not minimal. Now, assume $C$ is not minimal. Then there is a shortest
sequence of moves $\mSeq \in M \cup S$ such that $D = \mComp(C)$ and \(|D| <
|C|\). Since the sequence is shortest, and all $m \in M$ are volume-preserving,
$m_k$ must be associated with a sink $s \in S$ and $s \in D$.
\end{proof}

Although the characterization of minimality in \cref{minimal_vp} applies to all
volume-preserving models, its consequences vary. In the full tilt model it
implies that deciding minimality is \PSPACE-complete, whereas it is trivial in
the single step model, as we show in
\cref{prop:minimality-hard-ft,single_step_trivial}, respectively. In contrast to
this, the characterization of drainability in \cref{drainable_mono_vp} leads to a
polynomial-time decision procedure for the full tilt model, as we saw when
we investigated this special case in \cref{sec:drain-full-tilt}.

\begin{proposition}
    \label{prop:minimality-hard-ft}
    Given a configuration $C$ of a board $B$ and a set $S$ of sinks, it
    is \PSPACE-complete to decide whether $C$ is minimal with respect to $S$ in the
    full tilt model.
\end{proposition}
\begin{proof}
    This is a consequence of \cref{minimal_vp} and Theorem~5.1 in~\cite{hierarchical2020}, which states that the occupancy problem is \PSPACE-complete in the full tilt model.
\end{proof}

\begin{theorem}\label{drainable_mono_vp}
A board $B=(V,E)$ is drainable to a set of sinks $S$ in a monotone,
volume-preserving model $M$, if and only if, for every $p \in V$ there is \(s \in
S\) such that $p \rStar s$.
\end{theorem}
\begin{proof}
First, assume $B$ is drainable. Then, by \cref{drainable_emptyset}, no
configuration $\{p\}$ is minimal, for any $p \in V$. By \cref{minimal_vp}, and
since $M$ is volume-preserving, this means there is $s \in S$ such that \(p
\rStar s\).

Now, assume that for every $p \in V$ there is $s \in S$ with $p \rStar
s$. Assume, for sake of contradiction, that there is a minimal configuration \(C
\neq \varnothing\). Let $p \in C$ and $m \in M^*$ such that $m(\{p\}) = \{s\}$
for a sink $s \in S$. Then \(m(C) \supseteq m(\{p\}) = \{s\}\), due to
monotonicity. Thus, $C$ is not minimal by \cref{minimal_vp}---a
contradiction. Therefore, $\varnothing$ is the only minimal configuration, which
by \cref{drainable_emptyset} implies that $B$ is drainable.
\end{proof}

\begin{observation}\label{single_step_trivial}
Every board is drainable in the single step model, as long as every one of its
regions contains a sink.
\end{observation}
\begin{proof}
As every region contains a sink, there is a path on the board from every pixel
$p$ to a sink $s$. This path entails a sequence of single step moves
transforming $\{p\}$ into $\{s\}$. Thus, the board is drainable, by
\cref{obs:concrete-properties,drainable_mono_vp}.
\end{proof}

\begin{figure}[htb]
\centering
\begin{tikzpicture}[outer sep=auto]

\begin{scope}[pattern color=gray!30]
\fill [pattern=grid] (-4.4,0.7) -- (4.4,0.7) -- (4.4,-4.7) -- (0,-4.7) --
      (0,-3) -- (-2,-3) -- (-2,-1) -- (-4.4,-1) -- cycle;
\fill [pattern=crosshatch] (-4.4,-1.1) -- (-2.1,-1.1) -- (-2.1,-3.1) --
      (-0.1,-3.1) -- (-0.1,-4.7) -- (-4.4,-4.7) -- cycle;
\end{scope}

\begin{scope}[every node/.style={draw,align=center,minimum height=0.8cm,fill=white}]
\path (0,0) node [text width=8cm] (FT) {$\FT = \dual{\FT}$}
      (0,-2) node [text width=3cm] (FTI) {$\IE = \dual{\IE}$}
      (FT.west) ++(0,-2) node [right,text width=1.5cm] (S1) {$\SSt$}
      (FT.east) ++(0,-2) node [left,text width=1.5cm] (S1D) {$\dual{\SSt}$}
      (S1.west) ++(0,-2) node [right,text width=3.5cm] (S1I) {$\IE[\SSt]$}
      (S1D.east) ++(0,-2) node [left,text width=3.5cm] (S1ID) {$\dual{\IE[\SSt]}$};
\end{scope}

\begin{scope}[arrows=-{Stealth[length=2.5mm]},semithick]
\draw (S1.north) -- (S1.north |- FT.south);
\draw ($(FTI.north) - (0.15,0)$) -- ({$(FTI.north) - (0.15,0)$} |- FT.south);
\draw ($(S1D.north) - (0.15,0)$) -- ({$(S1D.north) - (0.15,0)$} |- FT.south);
\draw ({$(S1.south) - (0.15,0)$} |- S1ID.north) -- ($(S1.south) - (0.15,0)$);
\draw ($(S1I.north) + (1.2,0)$) -- ({$(S1I.north) + (1.2,0)$} |- FTI.south);
\draw ($(S1ID.north) - (1.2,0)$) -- ({$(S1ID.north) - (1.2,0)$} |- FTI.south);
\draw (S1D.south |- S1ID.north) -- (S1D.south);
\end{scope}

\begin{scope}[arrows=-{Stealth[length=2.5mm]},densely dashed,semithick]
\draw ($(S1.south) + (0.15,0)$) -- ({$(S1.south) + (0.15,0)$} |- S1ID.north);
\draw ({$(FTI.north) + (0.15,0)$} |- FT.south) -- ($(FTI.north) + (0.15,0)$);
\draw ({$(S1D.north) + (0.15,0)$} |- FT.south) -- ($(S1D.north) + (0.15,0)$);
\draw ({$(S1ID.north) - (0.25,0)$} |- FT.south) -- ($(S1ID.north) - (0.25,0)$);
\end{scope}

\end{tikzpicture}%
\caption{Overview of the studied models and their relationships. A solid arrow
  from  model $A$ to model $B$ indicates that $A$ simulates $B$; a dashed arrow
  indicates that $A$ simulates $B$ on singletons. Note that solid arrows include
  dashed arrows and arrows arising due to transitivity have been omitted. All
  models in the upper part (with the axis-parallel crosshatch pattern) are equivalent on singletons, as are the ones in the
  lower part (with the diagonal crosshatch pattern).
  Arrows on the left-hand side are derived in \cref{subsec:relative}; those on the right-hand side and the equalities come from~\cref{subsec:duality}.}
  \label{fig:models}
\end{figure}
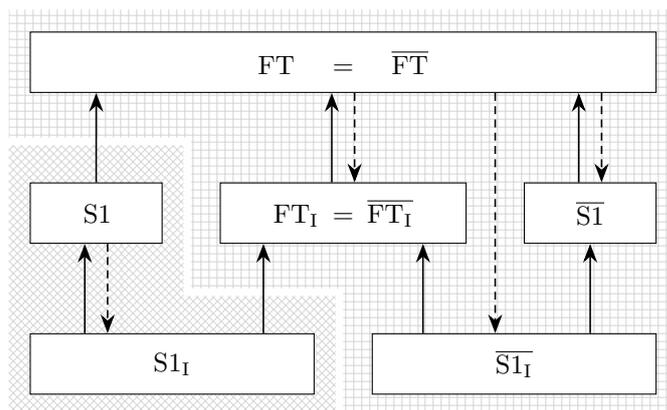

\subsection{Relative Power of Various Models}\label{subsec:relative}

Clearly, some models are more powerful than others, in the sense that they allow
us to reach more configurations. We make this notion precise by saying that a
model $M_2$ \emph{simulates} a model $M_1$ if for every $m \in M_1$ there is $m'
\in M_2^*$ such that $m(C) = m'(C)$ for every configuration $C$. It is easy to
see that \SSt simulates \FT (for any \(m \in \{u, d, \ell, r\}\), $m_\infty$
can be simulated using $D$ repetitions of $m_1$, where $D$ is the maximum
diameter among all regions of the board) and \IE[M] simulates $M$, for \(M \in
\{\FT, \SSt\}\) (choose the interval to encompass the whole board); see the left
half of \cref{fig:models}.

For the purpose of draining a board, it is useful to compare models with respect
to their moves acting on single particles. To this end, we say that model $M_2$
\emph{simulates $M_1$ on singletons} if, for every pixel $p$ and move \(m \in
M_1\), there is $m' \in M_2^*$ such that \(m(\{p\}) = m'(\{p\})\), i.e., if \(p
\rOne[M_1] q\) implies $p \rStar[M_2] q$, for all pixels $q$. Two models that
simulate each other on singletons are called \emph{equivalent on singletons}.
Note that general simulation entails simulation on singletons and that \SSt
and \FT simulate their respective interval extensions on singletons. We now
define a class of models with the property that all its members are equivalent
on singletons to \FT. It trivially includes \FT and \IE. Intuitively, these are
the models that move at least some particles maximally, and allow doing so in
all four directions.

\begin{definition}\label{def:tilt-compat}
A monotone and volume-preserving model $M$ is \emph{tilt-compatible} if the
following conditions are satisfied for all configurations $C$ and all occupied
pixels $p \in C$.
\begin{enumerate}
  \item For all $m \in M$, \(\{p^\ell, p^r, p^u, p^d, p\} \cap m(C) \neq
    \varnothing\).\label{tc_bound}
  \item For all \(x \in \{\ell, r, u, d\}\), there is $m \in M$ such that
    \(p^x \in m(C)\).\label{tc_choose}
\end{enumerate}
\end{definition}

\begin{proposition}\label{tc_equiv_ft}
Every tilt-compatible model is equivalent on singletons to \FT.
\end{proposition}
\begin{proof}
Let $M$ be a tilt-compatible model and $p$ a pixel. First, observe that
\(x_\infty(\{p\}) = \{p^x\}\), for all directions \(x \in \{u, d, \ell,
r\}\). Thus, by Condition~\ref{tc_choose} of \cref{def:tilt-compat} and the fact
that $M$ is volume-preserving, there is $m \in M$ with \(m(\{p\}) =
x_\infty(\{p\})\). Now, consider any $m \in M$. By Condition~\ref{tc_bound}, and
since $m$ is volume-preserving, $m(\{p\})$ is one of $\{p\}$, $\{p^\ell\}$,
$\{p^r\}$, $\{p^u\}$, or~$\{p^d\}$. In the first case, the empty sequence
$\varepsilon \in \FT^*$ satisfies \(\varepsilon(\{p\}) = \{p\}\); in the latter
cases \(x_\infty(\{p\}) = \{p^x\} = m (\{p\})\), for all directions \(x \in \{u,
d, \ell, r\}\).
\end{proof}

\begin{corollary}\label{cor:drainable-tc}
For every tilt-compatible model $M,$ a board is drainable to a set of sinks $S$
in $M$ if and only if it is drainable to $S$ in \FT.
\end{corollary}

\subsection{Duality of Various Models}\label{subsec:duality}

It can prove enlightening to imagine free pixels to instead be occupied by a
different kind of particle, called \emph{bubbles}, and analyze their behavior
under a sequence of moves. This leads to the concept of a \emph{dual move} for a
move $m$, defined as $\dual{m}(C) = V \setminus m(V \setminus C)$. In the
\emph{dual model} \(\dual{M} = \{\dual{m}: m \in M\}\) of a model $M,$ particles
move as bubbles do in $M,$ and vice versa; examples are depicted in~\cref{fig:duality}.

\begin{figure}[htb]
\centering
\includegraphics{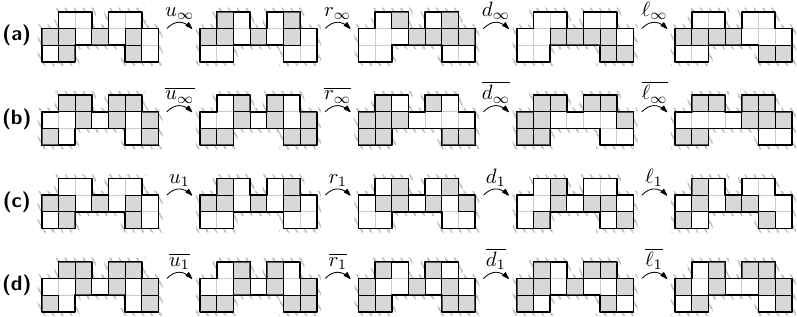}%
\caption{Particle movement in (a) \FT, (b) \dual{\FT}, (c) \SSt, and (d)
  \dual{\SSt}. Note that every move in \dual{\FT} corresponds to a move in \FT,
  whereas no such correspondence exists between \dual{\SSt} and
  \SSt.}\label{fig:duality}
\end{figure}

\begin{proposition}\label{prop:ft-dual}
\(\FT = \dual{\FT}\) and \(\IE = \dual{\IE}\).
\end{proposition}
\begin{proof}
We merely show $\ell = \dual{r}$, which easily extends to \(\ell^{[i,j]} =
\dual{r^{[i,j]}}\); a very similar argument applies to other pairs of opposite
directions. Let $C$ be a configuration, $R$ a row segment, $n=|R|$, and \(k=|R
\cap C|\). Then exactly the rightmost $n-k$ pixels of $R$ are occupied in \(r(V
\setminus C)\), and exactly the leftmost $k$ are occupied in \(V \setminus r(V
\setminus C) = \dual{r}(C)\)---the same ones that are occupied in $\ell(C)$.
\end{proof}

\begin{lemma}
    \label{dual_composition}
    For all moves $m_1$ and $m_2$, \(\dual{m_2 \circ m_1} = \dual{m_2} \circ
    \dual{m_1}\).
\end{lemma}
\begin{proof}
    Let $m_1$, $m_2$ be moves. Then,
    \begin{align*}
        \dual{m_2 \circ m_1}(C) &= V \setminus (m_2 \circ m_1(V \setminus C)) \\
        &= V \setminus m_2(V \setminus \dual{m_1}(C)) \\
        &= V \setminus (V \setminus \dual{m_2} \circ \dual{m_1}(C))\\
        &= \dual{m_2} \circ \dual{m_1}(C). \qedhere
    \end{align*}
\end{proof}

\begin{proposition}
    \label{dual_star_commutes}
    For every model $M$, \(\dual{M}^* = \dual{M^*}\).
\end{proposition}
\begin{proof}
    Use induction on the number of moves in a sequence and apply
    \cref{dual_composition}.
\end{proof}

\begin{restatable}{proposition}{restateDualSim}\label{dual_simulation}
If a model $M_2$ simulates a model $M_1$, then $\dual{M_2}$ simulates
$\dual{M_1}$.
\end{restatable}
\begin{proof}
    Let $C$ be a configuration and $m_1 \in \dual{M_1}$. Since $M_2$ simulates
    $M_1$, there is $m_2 \in M_2^*$ such that \(m_2(V \setminus C) = \dual{m_1}(V
    \setminus C)\). That is, \(\dual{m_2}(C) = V \setminus m_2(V \setminus C) = V
    \setminus \dual{m_1}(V \setminus C) = m_1(C)\). By \cref{dual_star_commutes},
    \(\dual{m_2} \in \dual{M_2}^*\). Therefore, \dual{M_2} simulates \dual{M_1}.
\end{proof}

We now come to the reason why we introduced the concept of tilt-compatible
models in the first place. It allows us to connect the full tilt model via
duality to the single step model and its interval extension. An overview of the
resulting relationships between the models in depicted in \cref{fig:models}.

\begin{proposition}\label{prop:concrete-duals-tc}
The models \dual{\FT}, \dual{\IE}, \dual{\SSt}, and \dual{\IE[\SSt]} are
tilt-compatible.
\end{proposition}
\begin{proof}
The claim is easy to see for \dual{\FT} and \dual{\IE} because they are dual to
themselves. For \dual{\SSt} observe the movement of bubbles in \SSt. Let $p$ be a
free pixel in a configuration $C$ and consider for every direction \(x \in \{u,
d, \ell, r\}\) the opposite direction $y$ ($d$, $u$, $r$, and $\ell$,
respectively). Then \(p^y \notin x_1(C)\), i.e., \(p^y \in \dual{x_1}(V
\setminus C)\). This satisfies condition~\ref{tc_choose} of
\cref{def:tilt-compat}. Since $\dual{u_1}$, $\dual{d_1}$, $\dual{\ell_1}$, and
$\dual{r_1}$ are the only moves of $\dual{\SSt}$, condition~\ref{tc_bound}
holds as well. \dual{\IE[\SSt]} simulates \dual{\SSt}, so
condition~\ref{tc_choose} is easily satisfied. For condition~\ref{tc_bound},
observe that particles move either as in $\dual{\SSt}$ or not at all.
\end{proof}

\subsection{Fillability Instead of Drainability}

Now, instead of removing as many particles as possible from a configuration, we
aim to insert as many as possible. A \emph{source} is a move associated with a
pixel $s \in V$ defined as \(C \mapsto V \cup \{s\}\). Note that the dual move
of a source at $s$ is a sink at $s$. Analogously to the notions of minimality
and drainability, a configuration $C$ is \emph{maximal} in a model $M$ with
respect to a set of sources $S$ if there is no $D$ reachable from $C$ in \(M
\cup S\) with \(|D| > |C|\), and a board $B=(V,E)$ is \emph{fillable} from $S$
if \(\varnothing \rStar[M \cup S] V\).

\begin{figure}[htb]
\centering
\includegraphics{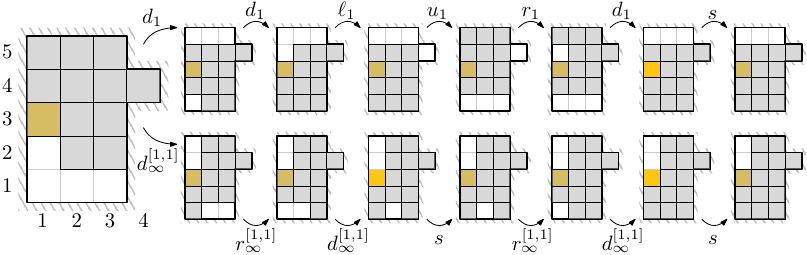}%
\caption{A configuration that is maximal with respect to a source $s$ in the
  full tilt model, and sequences of moves that lead to maximal configurations in
  \SSt (top row) and \IE (bottom row).
  A~sequence of moves leading to the initial configuration is
  \(m_1, r_\infty, m_1, r_\infty, m_1, m_2, \ell_\infty, m_2, \ell_\infty, m_2, u_\infty\),
  where $m_1$ is \(s, d_\infty, s, d_\infty, s\), and $m_2$ is
  \(u_\infty, r_\infty, d_\infty, s\).}\label{fig:maximal1}
\end{figure}

Due to \IE and \SSt simulating \FT, it is clear that both are at least as
powerful as \FT when it comes to reaching large maximal
configurations. \Cref{fig:maximal1} shows that they are both strictly more
powerful. In fact, \IE[\SSt] is more powerful still and could, in this example,
insert one more particle than \IE. This example would seem to indicate that \IE
is more powerful than \SSt but there are examples where the situation is
reversed, as illustrated in~\cref{fig:maximal2}.

\begin{figure}[htb]
	\centering
	\includegraphics{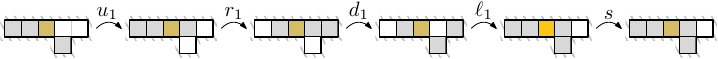}%
	\caption{A configuration that is maximal with respect to a source $s$ in \IE,
		and a sequence of moves that leads to a maximal configuration in
		\SSt.}\label{fig:maximal2}
\end{figure}

\begin{observation}\label{obs:minmax}
A configuration $C$ is maximal in a model $M$ with respect to $S$ if and only if
$V \setminus C$ is minimal in $\dual{M}$ with respect to $S$. Therefore, a
board is fillable from $S$ in $M$ if and only if it is drainable to $S$ in
$\dual{M}$.
\end{observation}

Finally, we come to the point when all the groundwork laid in the previous
subsections pays off. While drainability is trivial in \SSt,
fillability is not as obvious. Duality allows us to answer the question by
considering drainability in $\dual{\SSt}$, which we have connected to \FT.

\begin{theorem}\label{thm:tc-fill}
If the dual $\dual{M}$ of a model $M$ is tilt-compatible, then a board $B$ is
fillable from a set $S$ in $M$ if and only if $B$ is drainable to $S$ in \FT.
\end{theorem}
\begin{proof}
Follows from \cref{cor:drainable-tc,obs:minmax}.
\end{proof}

Consequently, the models \FT, \IE, \SSt, and \IE[\SSt] are all
equivalent with respect to fillability. Most importantly, the results regarding
drainability in the full tilt model presented in this paper, including hardness
and approximation for obstacle placement, apply to fillability in all these
models. We conclude by highlighting the most intriguing special case, which is
quite amusing when taken out of the context of the previous subsections.

\begin{corollary}\label{cor:s1-fill}
A board is fillable from a set $S$ in the single step model if and only if it is
drainable to $S$ in the full tilt model.
\end{corollary}

    \section{Conclusions and Future Work}\label{sec:conclusion}

In this paper, we have analyzed ways to make a board fillable or drainable using
particles moving as instructed by several models of global control signals. We
have shown that placing a minimum number of obstacles to
achieve fillability is \NP-hard in all considered models. However, a constant-factor approximation is
possible for scaled boards. The most apparent open question concerns the
complexity of the obstacle placement problem for scaled boards. Is it still
\NP-hard or can we do better than the approximation and solve it optimally?

The next step in our future work is to investigate how to actually fill a board,
once we know it is fillable. An immediate approach would be to maintain a
configuration in memory and successively move bubbles to a source until the board
is full. However, this may require an amount of memory that is not polynomial
in the size of an appropriate encoding of the boundary.
Can we efficiently compute (short) filling sequences? Of particular
interest is the worst case analysis of the length of filling sequences.

It would be interesting to better understand the geometric properties
of fillable regions. In the related problem of assembly there is a hierarchy of
constructable shapes~\cite{hierarchical2020}. These shapes derive from the
\emph{external} movement of particles, whereas ours derive from \emph{internal}
movement. Are these classifications related?

Our goal so far was to completely fill a region. Applications may
impose restricted areas, i.e., positions that may never be occupied by
particles, while requiring other areas to be filled. How can this constraint be
handled?

Although we introduced duality of models as a tool to connect fillability to
drainability, it leads to interesting new questions. Maximality in a model is
strongly related to the occupancy problem in the dual model. However,
occupiability in the dual model has a natural interpretation in the original
model, leading to what we dub the \emph{vacancy problem}: Given a configuration
$C$ and a pixel $p \in C$, is there a configuration $D$ reachable from $C$
that does not contain $p$?  This problem is \PSPACE-complete in the full tilt
model because so is the occupancy problem and \FT is dual to itself. To the best
of our knowledge, this natural problem has not been examined in the single step
model.

\begin{conjecture}\label{conj:vacancy}
The vacancy problem is \PSPACE-complete in the single step model.
\end{conjecture}

A promising approach to prove the conjecture is to follow the ideas
Caballero et~al.~\cite{caballero-cccg20-hardness} used to prove the relocation
problem hard for the single step model. They built on the techniques
of~\cite{hierarchical2020} and used an observation that can be seen as a
precursor to our notion of duality, implying that the same methods are likely
applicable in this case.
\newpage

    \bibliography{references}

\end{document}